\documentclass[11pt]{article}

\usepackage[a4paper,margin=27mm]{geometry}
\usepackage{amsmath,amssymb,amsthm,mathtools}
\usepackage{bm}
\usepackage{booktabs}
\usepackage{float}
\usepackage{graphicx}
\usepackage{authblk}
\usepackage{microtype}
\usepackage{xcolor}
\usepackage{hyperref}

\hypersetup{
  pdftitle={Robust CHSH Self-Testing with Finite-Energy GKP States},
  colorlinks=true,
  linkcolor=blue!55!black,
  citecolor=blue!55!black,
  urlcolor=blue!55!black
}

\newtheorem{theorem}{Theorem}[section]
\newtheorem{proposition}[theorem]{Proposition}
\newtheorem{lemma}[theorem]{Lemma}
\newtheorem{corollary}[theorem]{Corollary}

\newcommand{\cC}{\mathcal C}
\newcommand{\cH}{\mathcal H}
\newcommand{\Id}{\mathbb I}
\newcommand{\Tr}{\operatorname{Tr}}

\newcommand{\sgn}{\operatorname{sgn}}
\newcommand{\ket}[1]{\lvert #1\rangle}
\newcommand{\bra}[1]{\langle #1\rvert}

\newcommand{\ketbra}[2]{\lvert #1\rangle\!\langle #2\rvert}
\newcommand{\norm}[1]{\left\lVert #1\right\rVert}
\newcommand{\eps}{\varepsilon}
\newcommand{\betaCHSH}{\mathcal B}
\newcommand{\fK}{f_{\rm K}}
\title{Robust CHSH Self-Testing with Finite-Energy GKP States}
\author[1,2]{Farzin Salek}
\author[3,4,5]{Masahito Hayashi}
\affil[1]{Institute for Quantum Computing (IQC), University of Waterloo,
Ontario, Canada}
\affil[2]{Dahlem Center for Complex Quantum Systems, Freie Universit\"at
Berlin, Berlin, Germany}
\affil[3]{School of Data Science, The Chinese University of Hong Kong,
Shenzhen, Longgang District, Shenzhen 518172, China}
\affil[4]{International Quantum Academy, Futian District, Shenzhen 518048,
China}
\affil[5]{Graduate School of Mathematics, Nagoya University, Nagoya 464-8602,
Japan}
\date{August 2026}

\begin{document}

\maketitle

\begin{abstract}
We present a full-oscillator analysis of a finite-energy GKP CHSH test whose
observed score yields robust Bell-pair self-testing.  Periodically binned
position and momentum give the Pauli settings, while a fixed binary
coarse-graining of photon number modulo four and its displaced conjugate
realize the tilted settings.  For a number-filtered GKP source, we retain the
finite codeword overlap, define the measurements on all photon-number sectors,
and compute the physical correlations without logical post-corrections.  With
the canonical ideal-logical displacement \(d=\sqrt{\pi}\), the CHSH value
exceeds the local bound above \(4.56\) dB of per-peak squeezing, and
Kaniewski's extractability bound becomes nontrivial above \(5.02\) dB.
Calibrating only \(d\) using an independently characterized finite-energy
parameter lowers these model thresholds to \(4.21\) dB and \(4.58\) dB,
respectively; at \(12\) dB, it raises the score from \(2.69486\) to \(2.78858\)
and the corresponding target-state overlap bound from \(0.90758\) to
\(0.97243\).  This calibration is fixed before Bell-test data are collected.
The displacement activates the odd modulo-four sectors, so their fixed
a priori assignments are a genuine finite-energy component.  The large gain
is specific to the deterministic phase-bit coarse-graining; independently
calibrating the one-bit POVM with randomized odd-sector outcomes gives only a
much smaller improvement.  These are honest-model predictions, not loss,
detection-efficiency, or finite-sample thresholds.  In an experiment, a
device-independent guarantee for an extracted Bell pair follows by inserting
a confidence lower bound on the observed CHSH score into the self-testing
theorem.
\end{abstract}

\section{Introduction}
\label{sec:introduction}

Self-testing turns observed correlations into a characterization of an
unknown quantum state and its measurements, up to local extraction
maps~\cite{MayersYao2004,McKague2012,Kaniewski2016,SupicBowles2020}.
The conclusion is dimension independent: the physical systems need not be
qubits, provided that their observed statistics satisfy the hypotheses of the
self-testing theorem.  Here robustness means the quantitative conversion of
a nonideal CHSH score into a Bell-pair extractability guarantee; it does not
mean a derived tolerance against loss, detector inefficiency, or finite sample
size.  This makes bosonic encodings natural candidates for continuous-variable
implementations of device-independent protocols.
Bosonic quantum error-correcting codes encode finite-dimensional information
in oscillator degrees of freedom and offer several routes to hardware-efficient
protection~\cite{TerhalConradVuillot2020,Albert2025}.

The Gottesman--Kitaev--Preskill (GKP) code embeds a qubit in one oscillator by
using a grid in phase space~\cite{GKP2001}.  Position and momentum homodyne
measurements, followed by a fixed periodic binning, implement the logical
\(Z\)- and \(X\)-readouts in the ideal limit.  This observation underlies an
early proposal for device-independent quantum cryptography with GKP-encoded
Bell states~\cite{MarshallWeedbrook2014}.  A bipartite Bell test, however,
requires measurement directions outside the logical Pauli axes.  In the
proposal of Ref.~\cite{MarshallWeedbrook2014}, the corresponding changes of
basis use a logical \(T\) gate and hence a non-Gaussian resource.  More
recently, Yang \emph{et al.} proved that the homodyne-only palette
\(X,Y,Z\) cannot violate CHSH on an encoded Bell pair, even though the same
palette can reveal multipartite nonlocality~\cite{Yang2026}.  Thus the
tilted measurements are the central physical bottleneck in the bipartite
setting.
This does not rule out homodyne Bell tests for other states: recent work by
Lopetegui-Gonz\'alez \emph{et al.} instead tailors the state to the available
continuous-variable measurements~\cite{Lopetegui2026}.  Here we retain the
encoded Bell-pair target and change the measurement resource.

A tempting response is to homodyne the diagonal quadratures.  This does not
give \((X\pm Z)/\sqrt2\).  The Weyl phase instead makes the two diagonal
readouts logical \(\pm Y\) measurements.  Rotating the resulting compressed
\(2\times2\) matrix to a desired Bloch-sphere direction does not solve the
physical problem: unless the conjugating oscillator unitary is specified, the
procedure only changes coordinates inside a model of the code space.

Here we use a different property of the square GKP code.  The phase-fixed
quarter rotation
\[
 F=e^{i\pi\hat n/2}
\]
acts as the logical Hadamard.  This relation, and its connection to
photon-number-modulo-four readout, already appear in the original GKP
proposal~\cite{GKP2001}; later cavity constructions make the same symmetry
operational~\cite{TerhalWeigand2016}.  Consequently, photon number modulo four
contains a direct readout of the tilted observable
\((X+Z)/\sqrt2\).  A displacement conjugation gives the other tilted
observable in the ideal code.  The important finite-energy questions are not
answered by the ideal identity: the displacement does not commute with the
energy filter, it creates odd-parity support, and a binary measurement must be
defined on all four residue classes of photon number.  These effects must be
calculated on the oscillator, not imposed after compression.

The quarter-rotation identity and its modular-photon-number interpretation are
known, as is diagonal homodyne as a logical-\(Y\) readout.  The contribution of
the present work is their full-oscillator finite-energy incorporation into a
CHSH self-testing construction.  In particular, we define the binary
measurements on every photon-number sector, evaluate the displaced setting
directly without postselection or an abstract logical correction, retain the
finite codeword overlap, and obtain quantitative Bell-violation and
extractability thresholds.  To our knowledge, this is the first
full-oscillator finite-energy analysis of a GKP CHSH construction whose tilted
settings are realized using modular photon-number information.

\begin{enumerate}
  \item We use the number-damped family
  \(e^{-\beta\hat n}\ket{j_{\rm GKP}}\) throughout.  Its position-space form,
  lattice contraction, peak width, and squeezing parameter are related
  exactly~\cite{Matsuura2020}.
  \item We define full-Hilbert-space reflections for the four measurements in
  the CHSH test: periodically binned \(q\) and \(p\), a fixed binary
  coarse-graining of \(\hat n\bmod4\), and the displaced conjugate of that
  coarse-graining.  No postselection or code-space-only correction is used.
  \item We calculate the physical correlations of the normalized, filtered
  Bell state.  The nonorthogonality of the filtered codewords is retained
  through their Gram matrix.  A canonical polar isometry is used only to
  describe the logical action of the measurements, not as a replacement for
  them.
  \item We evaluate the analytic extractability curve of
  Kaniewski~\cite{Kaniewski2016} on the directly calculated CHSH value.  The
  same bound gives a rigorous device-independent lower bound when supplied
  with an experimentally valid confidence lower bound.  We also report the
  correlations relevant to the larger six-setting test used in graph-state verification
  protocols~\cite{HayashiHajdusek2018}.
  \item Keeping the source, homodyne bins, and odd-sector rule fixed, we
  calibrate only the displacement amplitude in the second modular-four
  setting.  This isolates the finite-energy gain without turning the
  construction into a multi-parameter fit.  The deterministic odd-sector
  assignment is fixed throughout, so the reported large calibrated gain belongs to this
  declared physical coarse-graining and is not obtained by optimizing the
  residue-class labels.
\end{enumerate}

For the canonical displacement \(d=\sqrt\pi\), Bell violation starts at
approximately \(4.56\) dB and the Kaniewski-bound fidelity estimate becomes
nontrivial at approximately \(5.02\) dB.  Finite-energy displacement
calibration lowers these crossings to \(4.21\) dB and \(4.58\) dB.  The main
finite-energy penalty comes from the displaced modular-photon-number setting:
the undisplaced setting is an exact Hadamard on the canonically embedded
finite-energy code, whereas a displacement does not preserve that code
manifold.

The finite-energy oscillator calculation and the device-independent
certification play different roles.  The former predicts the CHSH score
produced by the specified GKP state and physical measurements under an
honest-device model.  In an experiment, however, the self-testing guarantee is
obtained by inserting a confidence lower bound on the observed black-box CHSH
score into Corollary~\ref{thm:kaniewski}.  That certification step depends only
on the observed score and does not assume the GKP model.

The remainder of the paper is organized as follows.
Section~\ref{sec:chsh} fixes the target and states the quantitative CHSH
self-test.  Section~\ref{sec:finite-energy} defines the finite-energy GKP
family and its canonical logical frame.
Section~\ref{sec:measurements} gives the four full-oscillator measurements and
the displacement-calibration prescription.
Section~\ref{sec:correlations} derives the finite-dimensional formulas used
for the physical correlations, and Section~\ref{sec:numerics} reports the
results and convergence tests.  Experimental scope and extensions are
discussed in Section~\ref{sec:discussion}.

\section{Target correlations and quantitative CHSH self-testing}
\label{sec:chsh}

\subsection{Target state and settings}

We use the Bell-like state
\begin{equation}
 \ket{\psi_{\rm B}}
 :=
 \frac{\ket{0}\ket{+}+\ket{1}\ket{-}}{\sqrt2}
 =
 (\Id\otimes H)\ket{\Phi^+},
 \label{eq:target-state}
\end{equation}
where
\[
 \ket{\pm}:=\frac{\ket0\pm\ket1}{\sqrt2},
 \qquad
 \ket{\Phi^+}:=\frac{\ket{00}+\ket{11}}{\sqrt2},
 \qquad
 H:=\frac{X+Z}{\sqrt2}.
\]
Alice's two CHSH observables are
\begin{equation}
 A_+:=\frac{X+Z}{\sqrt2},
 \qquad
 A_-:=\frac{X-Z}{\sqrt2},
 \label{eq:ideal-alice}
\end{equation}
and Bob's are \(Z\) and \(X\).  With the sign convention used below, the
Clauser--Horne--Shimony--Holt (CHSH) operator~\cite{CHSH1969} is
\begin{equation}
 \betaCHSH
 :=
 A_+\otimes(Z+X)+A_-\otimes(Z-X).
 \label{eq:ideal-chsh}
\end{equation}
The four ideal correlators are
\begin{equation}
 \langle A_+\otimes Z\rangle=\langle A_+\otimes X\rangle=\frac1{\sqrt2},
 \qquad
 \langle A_-\otimes Z\rangle=-\langle A_-\otimes X\rangle=\frac1{\sqrt2}.
 \label{eq:ideal-chsh-correlators}
\end{equation}
Thus the target attains
\begin{equation}
 \bra{\psi_{\rm B}}\betaCHSH\ket{\psi_{\rm B}}=2\sqrt2.
\end{equation}

The implementation also includes physical \(X\)- and \(Z\)-type measurements
on Alice.  They provide the additional ideal checks
\begin{equation}
 \langle X\otimes Z\rangle=1,
 \qquad
 \langle Z\otimes X\rangle=1,
 \qquad
 \langle X\otimes X+Z\otimes Z\rangle=0.
 \label{eq:extra-ideal-checks}
\end{equation}
These checks are not needed for the CHSH extractability bound below, but they
connect the construction to six-setting Bell-pair tests used in
verification protocols~\cite{HayashiHajdusek2018}.

\subsection{An explicit extraction-fidelity bound}

For a bipartite state \(\rho\), let \(S\) denote its observed expectation of a
CHSH operator built from binary observables.  Its extractability to the target
is
\begin{equation}
 \Xi(\rho\to\psi_{\rm B})
 :=
 \sup_{\Lambda_A,\Lambda_B}
 \bra{\psi_{\rm B}}
 (\Lambda_A\otimes\Lambda_B)(\rho)
 \ket{\psi_{\rm B}},
 \label{eq:pure-target-overlap}
\end{equation}
where the supremum is over local normal completely positive trace-preserving
maps with qubit outputs.  The following is the specialization of the analytic
bound in Ref.~\cite{Kaniewski2016} to the Bell state in
Eq.~\eqref{eq:target-state}.  Kaniewski's canonical maximally entangled target
is locally unitarily equivalent to \(\ket{\psi_{\rm B}}\); postcomposing one
extraction channel with the corresponding fixed output unitary leaves the
bound unchanged.  In our convention,
\(\ket{\psi_{\rm B}}=(\Id\otimes H)\ket{\Phi^+}\).
The qubits in Eq.~\eqref{eq:pure-target-overlap} are output registers of the
local extraction channels.  No two-dimensional subspace or direct-sum sector
of either physical Hilbert space is assumed.  In particular, for oscillator
inputs the extracted qubit need not coincide with the model-specific GKP code
frame introduced below.

\begin{theorem}[Finite-dimensional CHSH extractability bound~\cite{Kaniewski2016}]
\label{thm:kaniewski-finite}
For every finite-dimensional bipartite state \(\rho\) and two binary
reflections per party with CHSH expectation \(S\in[2,2\sqrt2]\), the
extractability to a maximally entangled two-qubit state is at least
\(\fK(S)\).  After the fixed local output unitary relating Kaniewski's target
to \(\ket{\psi_{\rm B}}\), this reads
\begin{equation}
 \Xi(\rho\to\psi_{\rm B})\ge\fK(S),
 \label{eq:kaniewski-bound}
\end{equation}
where
\begin{equation}
 \fK(S):=\max\left\{\frac12,
 \frac{4+5\sqrt2}{16}S-\frac{1+2\sqrt2}{4}\right\}.
 \label{eq:fK}
\end{equation}
The affine part exceeds \(1/2\) at
\(S_*=(16+14\sqrt2)/17\simeq2.10582\) and reaches one at
\(S=2\sqrt2\).
\end{theorem}
The reflection formulation also covers two binary POVMs per party.  On one
common finite-dimensional dilation space for each party, both POVMs are
represented by reflections using the same local isometric embedding.  The
embedded state then preserves all four CHSH correlators, and the extraction
channel on the dilation may be precomposed with that embedding.  The explicit
simultaneous construction used below is recorded in
Appendix~\ref{app:infinite-dimension}.

\begin{corollary}[Separable-Hilbert-space CHSH extractability bound]
\label{thm:kaniewski}
Let \(\rho\) be a normal state on \(\cH_A\otimes\cH_B\), where both spaces
are separable, and let each party have two binary POVMs.  If the associated
CHSH expectation is \(S\in[2,2\sqrt2]\), then
\begin{equation}
 \Xi(\rho\to\psi_{\rm B})\ge\fK(S).
 \label{eq:kaniewski-separable-bound}
\end{equation}
Equivalently, for every \(\eps>0\), there are local normal CPTP maps with
qubit outputs whose output overlap with \(\ket{\psi_{\rm B}}\) is at least
\(\fK(S)-\eps\).
\end{corollary}
Appendix~\ref{app:infinite-dimension} proves
Corollary~\ref{thm:kaniewski} by finite-rank approximation.  The approximation
is only a mathematical proof device: it neither postselects the Bell
experiment nor identifies an oscillator subspace with the extracted qubit.
Here fidelity with the pure target means the overlap in
Eq.~\eqref{eq:pure-target-overlap}, not its square root.  If finite data give a
statistically valid lower confidence bound
\(S_{\rm L}\in[2,2\sqrt2]\), then \(\fK(S_{\rm L})\) is the corresponding
confidence-level extractability lower bound within the standard tensor-product
quantum model.  Constructing such a finite-sample bound is outside the scope
of this paper.  If \(S_{\rm L}<2\), the affine branch is not invoked and
constant local output channels still guarantee overlap \(1/2\).

\section{One finite-energy GKP convention}
\label{sec:finite-energy}

\subsection{Ideal code and phase convention}

We set \(\hbar=1\) and
\[
 [\hat q,\hat p]=i,
 \qquad
 \hat a=\frac{\hat q+i\hat p}{\sqrt2},
 \qquad
 \hat n=\hat a^\dagger\hat a.
\]
The centered square-GKP codewords are the formal distributions
\begin{equation}
 \ket{\bar j}
 \propto
 \sum_{s\in\mathbb Z}
 \ket{(2s+j)\sqrt\pi}_{q},
 \qquad j\in\{0,1\}.
 \label{eq:ideal-gkp}
\end{equation}
The two distributions in Eq.~\eqref{eq:ideal-gkp} use one common positive
distributional scale and zero relative phase.  All exact Fourier identities
below refer to this common convention.
The logical displacement convention is
\begin{equation}
 \bar X=e^{-i\sqrt\pi\,\hat p},
 \qquad
 \bar Z=e^{i\sqrt\pi\,\hat q}.
 \label{eq:logical-displacements}
\end{equation}
The centered codewords have even photon-number parity.

Define the phase-fixed quarter rotation
\begin{equation}
 F:=e^{i\pi\hat n/2}.
 \label{eq:F}
\end{equation}
It satisfies
\begin{equation}
 F\hat qF^\dagger=\hat p,
 \qquad
 F\hat pF^\dagger=-\hat q.
 \label{eq:F-quadratures}
\end{equation}
The phase in Eq.~\eqref{eq:F} is essential.  The metaplectic rotation
\[
 R(\pi/2)
 =
 e^{i\pi(\hat q^2+\hat p^2)/4}
 =
 e^{i\pi/4}F
\]
has the same action by conjugation, but its restriction to the code is
\(e^{i\pi/4}H\), not a Hermitian logical reflection.  In the Fourier
convention fixed by
\({}_{q}\langle q|p\rangle=(2\pi)^{-1/2}e^{iqp}\), Poisson summation gives
\begin{equation}
 F\ket{\bar0}
 =
 \frac{\ket{\bar0}+\ket{\bar1}}{\sqrt2},
 \qquad
 F\ket{\bar1}
 =
 \frac{\ket{\bar0}-\ket{\bar1}}{\sqrt2}.
 \label{eq:F-logical-H}
\end{equation}
Thus \(F\) restricts to \(+H\), with the outcome labels fixed.  We include
the short distributional derivation in Appendix~\ref{app:phase}.
The identity itself is known~\cite{GKP2001}; the explicit phase check is
included because the sign becomes observable in the modular-number
coarse-graining below.

\subsection{Number damping and exact parameter conversion}

We use the isotropic energy filter
\begin{equation}
 E_\beta:=e^{-\beta\hat n},
 \qquad \beta>0,
 \label{eq:energy-filter}
\end{equation}
and define the two unnormalized filtered columns
\begin{equation}
 \ket{w_{j,\beta}}:=E_\beta\ket{\bar j},
 \qquad
 W_\beta:=\sum_{j=0}^1\ket{w_{j,\beta}}\bra j.
 \label{eq:Wbeta}
\end{equation}
Multiplying \(E_\beta\) by \(e^{-\beta/2}\) does not change any normalized
state.  The heat kernel of \(e^{-\beta(\hat n+1/2)}\) gives
\begin{align}
 {}_q\langle q|w_{j,\beta}\rangle
 \propto
 \sum_{s\in\mathbb Z}
 &\exp\left[-\frac{\tanh\beta}{2}
              (2s+j)^2\pi\right]
 \nonumber\\[-1mm]
 &\times
 \exp\left[
  -\frac{\left(q-(2s+j)\sqrt\pi/\cosh\beta\right)^2}
        {2\tanh\beta}
 \right].
 \label{eq:number-filtered-wavefunction}
\end{align}
This exact formula displays both the envelope and the contracted lattice
centers~\cite{Matsuura2020}.
\begin{lemma}[Normalizability after number damping]
For every \(\beta>0\), the formal distributional action of \(E_\beta\) on each ideal comb defines an \(L^2(\mathbb R)\) vector.
\end{lemma}
\begin{proof}
The coefficient of the \(s\)-th translated Gaussian is bounded by \(e^{-c_\beta(2s+j)^2}\) for some \(c_\beta>0\), while all translated peaks have the same finite \(L^2\) norm.  The sum of these norms is finite, so the Gaussian-comb series converges absolutely in \(L^2\).  Applying the Mehler kernel to finite comb truncations and taking this limit yields Eq.~\eqref{eq:number-filtered-wavefunction} and justifies the Fock expansion below.
\end{proof}

In the common symmetric Gaussian-comb convention
used in that reference, the number-filter parameter and comb parameter are
related by
\begin{equation}
 \Delta^2=\tanh\beta.
 \label{eq:Delta-beta}
\end{equation}
The lattice-center contraction in
Eq.~\eqref{eq:number-filtered-wavefunction} is therefore
\begin{equation}
 \frac{1}{\cosh\beta}=\sqrt{1-\Delta^4},
 \qquad
 \beta=\operatorname{arctanh}(\Delta^2)
      =\Delta^2+O(\Delta^6).
 \label{eq:contraction-Delta}
\end{equation}
Thus the often-written comb with uncontracted centers is recovered at high
squeezing.  We use the per-peak convention
\begin{equation}
 r_{\rm dB}:=-10\log_{10}(\Delta^2)
             =-10\log_{10}(\tanh\beta)
             =-20\log_{10}\Delta
 \label{eq:squeezing-conversion}
\end{equation}
to report squeezing.  The probability variance of an individual peak is
\(\Delta^2/2\); for orientation, \(\Delta=0.3\) corresponds to
\(10.46\) dB in this convention.  Stabilizer-derived effective squeezing is a
different diagnostic for a generic nonideal state and should not be identified
with the exact parameter conversion in Eq.~\eqref{eq:Delta-beta}.

\subsection{The filtered Bell state and the canonical logical frame}

The physical two-mode state studied in this paper is the normalized filtered
image of the ideal target:
\begin{equation}
 \ket{\Psi_\beta}
 :=
 \frac{
 (W_\beta\otimes W_\beta)\ket{\psi_{\rm B}}
 }{\sqrt{\mathcal N_\beta}},
 \qquad
 \mathcal N_\beta
 :=
 \bra{\psi_{\rm B}}
 (S_\beta\otimes S_\beta)
 \ket{\psi_{\rm B}},
 \label{eq:filtered-bell}
\end{equation}
where
\begin{equation}
 S_\beta:=W_\beta^\dagger W_\beta.
 \label{eq:gram}
\end{equation}
The two columns in Eq.~\eqref{eq:Wbeta} are not separately normalized before
filtering the Bell state.  Equation~\eqref{eq:filtered-bell} therefore defines the normalized source-state family obtained from the formal action of the same number filter on the ideal two-mode code state.
It is the assumed honest-source model in this work.  We use the normalized
vector in Eq.~\eqref{eq:filtered-bell} as a source-state ansatz.  This
normalization does not represent postselection on measurement outcomes within
the Bell test, and we do not model the preparation probability or a physical
implementation of the filter.  We do not claim that a particular noisy
preparation circuit produces this state exactly.

The ideal combs are linearly independent because their position supports
occupy distinct lattice cosets.  Moreover, \(E_\beta\) multiplies every Fock
coefficient by the nonzero number \(e^{-\beta n}\), so it is injective on
their two-dimensional distributional span.  The filtered columns therefore
remain linearly independent and \(S_\beta>0\).
For logical diagnostics, the unique canonical polar embedding associated with
\(W_\beta\) is
\begin{equation}
 V_\beta:=W_\beta S_\beta^{-1/2}.
 \label{eq:canonical-isometry}
\end{equation}
It obeys \(V_\beta^\dagger V_\beta=\Id_2\).  A physical oscillator operator
\(M\) has the unambiguous canonical compression
\begin{equation}
 M_\beta^{\rm log}:=V_\beta^\dagger M V_\beta.
 \label{eq:canonical-compression}
\end{equation}
This compression is used to interpret a physical measurement.  The
self-testing correlations themselves will instead be evaluated directly in
\(\ket{\Psi_\beta}\).  This canonical two-dimensional frame is used only to
interpret the honest GKP model and is not identified with the abstract qubit
output of the self-testing extraction channels.

\begin{proposition}[Exact finite-energy Hadamard covariance]
\label{prop:exact-H}
For the number-filtered family,
\begin{equation}
 FW_\beta=W_\beta H,
 \qquad
 [S_\beta,H]=0,
 \qquad
 FV_\beta=V_\beta H.
 \label{eq:finite-H-covariance}
\end{equation}
\end{proposition}

\begin{proof}
Equations~\eqref{eq:energy-filter} and~\eqref{eq:F} are functions of
\(\hat n\), so they commute.  Combining this with
Eq.~\eqref{eq:F-logical-H} yields \(FW_\beta=W_\beta H\).  Unitarity of
\(F\) then gives
\[
 S_\beta=(FW_\beta)^\dagger(FW_\beta)
         =H^\dagger S_\beta H.
\]
Since \(H=H^\dagger=H^{-1}\), this is equivalent to
\([S_\beta,H]=0\).  The last identity follows from
\(FV_\beta=W_\beta H S_\beta^{-1/2}
=W_\beta S_\beta^{-1/2}H\).
\end{proof}

The filtered state in Eq.~\eqref{eq:filtered-bell} is generally not equal to
\((V_\beta\otimes V_\beta)\ket{\psi_{\rm B}}\).  Indeed,
\[
 W_\beta=V_\beta S_\beta^{1/2},
\]
so the physical filter also produces the logical distortion
\(S_\beta^{1/2}\otimes S_\beta^{1/2}\).  This distinction is retained in all
numerical correlations.

\section{Full-oscillator measurement settings}
\label{sec:measurements}

\subsection{Position and momentum binnings}

Let
\begin{equation}
 b(x):=\sgn\!\left[\cos(\sqrt\pi\,x)\right],
 \label{eq:bin-function}
\end{equation}
with either fixed assignment at the measure-zero bin boundaries.  We define
the physical reflections
\begin{equation}
 Q:=b(\hat q),
 \qquad
 P:=b(\hat p)=FQF^\dagger.
 \label{eq:QP}
\end{equation}
Both are Hermitian and satisfy
\begin{equation}
 Q^2=P^2=\Id_{\cH_{\rm osc}}.
\end{equation}
On the ideal centered code, \(Q\) and \(P\) restrict to \(Z\) and \(X\),
respectively.  They are not exact logical Pauli operators for finite
\(\beta\); their physical correlations are evaluated below.  We keep the
fixed ideal-lattice binning in Eq.~\eqref{eq:bin-function}, rather than
recalibrating the bins as a function of \(\beta\).  This matches the standard
periodic readout and keeps the measurement definition independent of the
unknown state~\cite{Yang2026}.

\subsection{Why diagonal homodyne does not give a tilted real axis}

The diagonal quadratures provide a useful diagnosis of an otherwise plausible
but incorrect construction.  With the convention
\(Y=iXZ\), the Baker--Campbell--Hausdorff formula gives, on the ideal code,
\begin{align}
 e^{i\sqrt\pi(\hat q+\hat p)}
 &=
 i\,\bar Z\bar X^\dagger
 \longmapsto -Y,
 \label{eq:qplusp-Y}\\
 e^{i\sqrt\pi(\hat q-\hat p)}
 &=
 -i\,\bar Z\bar X
 \longmapsto +Y.
 \label{eq:qminusp-Y}
\end{align}
Accordingly, periodic binnings aligned with the two diagonal grid directions
read out \(-Y\) and \(+Y\), up to the corresponding convention for outcome
labels.  This is the expected GKP logical-\(Y\) readout, not a finite-energy
misalignment of \((X\pm Z)/\sqrt2\).  Diagonal probability densities alone do
not expose the phase in Eqs.~\eqref{eq:qplusp-Y} and
\eqref{eq:qminusp-Y}; the off-diagonal code-space matrix elements are needed
to distinguish an \(X\)-type action from a \(Y\)-type action.  We therefore do
not use diagonal homodyne in the CHSH settings.

\subsection{A projective modular-four reflection}

Let
\begin{equation}
 \Pi_r
 :=
 \sum_{k=0}^{\infty}
 \ketbra{4k+r}{4k+r},
 \qquad r=0,1,2,3.
 \label{eq:Pi-r}
\end{equation}
The quarter rotation has spectral decomposition
\begin{equation}
 F=\Pi_0+i\Pi_1-\Pi_2-i\Pi_3.
 \label{eq:F-spectral}
\end{equation}
A four-outcome measurement of \(\hat n\bmod4\) is not yet the binary
measurement required in a CHSH experiment.  We fix the deterministic
coarse-graining
\begin{equation}
 M_4:=\Pi_0+\Pi_1-\Pi_2-\Pi_3.
 \label{eq:M4}
\end{equation}
Thus \(\{0,1\}\) gives outcome \(+1\) and \(\{2,3\}\) gives outcome \(-1\).
Equivalently,
\begin{equation}
 M_4
 =
 \sqrt2\cos\left[\frac{\pi}{2}
 \left(\hat n-\frac12\right)\right]
 =
 \frac{e^{-i\pi/4}F+e^{i\pi/4}F^\dagger}{\sqrt2}.
 \label{eq:M4-function}
\end{equation}
Label the resolved bits by \(b_1b_0\) and set \(r=2b_1+b_0\).  Then
\[\begin{array}{c|c|c|c}r&F\text{ eigenvalue}&b_1b_0&M_4\text{ output}\\\hline 0&1&00&+1\\1&i&01&+1\\2&-1&10&-1\\3&-i&11&-1\end{array}\]
Thus \(M_4=(-1)^{b_1}\).  Exchanging ancilla order or phase convention merely relabels the resolved projectors if the classical output map is relabeled with them.  We fix this table and \(F=e^{i\pi\hat n/2}\); in this precise sense, \(M_4\) is the most-significant phase bit.
It is a full-Hilbert-space reflection:
\begin{equation}
 M_4^\dagger=M_4,
 \qquad
 M_4^2=\Id_{\cH_{\rm osc}}.
 \label{eq:M4-reflection}
\end{equation}

The two tilted physical settings are
\begin{equation}
 A_+^{\rm phys}:=M_4,
 \qquad
 A_-^{\rm phys}:=
 \bar X^\dagger M_4\bar X,
 \qquad
 \bar X=e^{-i\sqrt\pi\hat p}.
 \label{eq:physical-tilted}
\end{equation}
Both are Hermitian reflections on the entire oscillator Hilbert space.

\begin{proposition}[Logical action of the tilted settings]
\label{prop:tilted-action}
The first setting is an exact Hadamard on the canonically embedded
number-filtered code:
\begin{equation}
 M_4V_\beta=V_\beta H,
 \qquad
 V_\beta^\dagger M_4V_\beta
 =
 H=\frac{X+Z}{\sqrt2}.
 \label{eq:M4-exact-H}
\end{equation}
In the ideal code,
\begin{equation}
 \left.
 \bar X^\dagger M_4\bar X
 \right|_{\cC_{\rm ideal}}
 =
 XHX
 =
 \frac{X-Z}{\sqrt2}.
 \label{eq:Aminus-ideal}
\end{equation}
Unlike Eq.~\eqref{eq:M4-exact-H}, Eq.~\eqref{eq:Aminus-ideal} is asserted only
for the ideal code; at finite \(\beta\), the conjugated observable is evaluated
on the full oscillator state below.
\end{proposition}

\begin{proof}
Every vector in the centered filtered code has even parity.  On the even subspace, both \(M_4\) and \(F\) equal \(+1\) on \(n=0\pmod4\) and \(-1\) on \(n=2\pmod4\).  Hence \(M_4=F\) on the range of \(V_\beta\), and
Eq.~\eqref{eq:M4-exact-H} follows from
Proposition~\ref{prop:exact-H}.  On the ideal code,
\(\bar X\) restricts to logical \(X=X^\dagger\), which gives
Eq.~\eqref{eq:Aminus-ideal}.  At finite energy,
\([E_\beta,\bar X]\neq0\).  The parity calculation below explicitly tracks
the odd-sector support generated by the displacement, so its conjugated action
must be evaluated on the full oscillator state.
\end{proof}

The relevant finite-energy mechanism can be seen directly from parity.  Let
\begin{equation}
 \mathcal P:=(-1)^{\hat n}=F^2,
 \qquad
 \Pi_{\rm odd}:=\Pi_1+\Pi_3=\frac{\Id-\mathcal P}{2}.
 \label{eq:parity-odd-projector}
\end{equation}
The centered ideal codewords are parity even, and the number filter commutes
with parity.  Hence
\begin{equation}
 \mathcal P W_\beta=W_\beta.
 \label{eq:filtered-even}
\end{equation}
Parity reverses the displacement,
\begin{equation}
 \mathcal P\bar X\mathcal P=\bar X^\dagger.
 \label{eq:parity-displacement}
\end{equation}
Consequently, for every
\(\ket{\psi_\beta}\in\operatorname{ran}W_\beta\),
\begin{equation}
 \Pi_{\rm odd}\bar X\ket{\psi_\beta}
 =
 \frac{\bar X-\bar X^\dagger}{2}\ket{\psi_\beta}.
 \label{eq:odd-component}
\end{equation}
This vector is generally nonzero at finite energy.  On the ideal code,
\(\bar X^2\) is a stabilizer, so \(\bar X\) and \(\bar X^\dagger\) have the
same logical action and the odd component vanishes formally.  Thus the key
fact is not only \([E_\beta,\bar X]\neq0\): the displacement fails to preserve
parity and therefore fails to preserve the centered finite-energy code
manifold.  The second modular measurement genuinely probes the
\(n=1,3\pmod4\) sectors, whose binary assignments may not be discarded.

\subsection{Odd-sector extensions and the one-bit alternative}

Every deterministic binary coarse-graining of the fixed modular-four PVM
\(\{\Pi_r\}_{r=0}^3\) that assigns the logical Hadamard outcomes on the even
sectors has the form
\begin{equation}
 M_4^{(s_1,s_3)}
 :=
 \Pi_0+s_1\Pi_1-\Pi_2+s_3\Pi_3,
 \qquad
 s_1,s_3\in\{\pm1\}.
 \label{eq:all-extensions}
\end{equation}
The ideal code does not select \(s_1,s_3\).  Our primary choice,
Eq.~\eqref{eq:M4}, is \((s_1,s_3)=(+1,-1)\) and is motivated by the
binary phase bit.  We nevertheless evaluate all four choices in
Section~\ref{subsec:extension-sensitivity}.

This dependence can also be bounded directly.  Define
\begin{equation}
 \ell_{\rm odd}(\beta)
 :=
 \norm{
 V_\beta^\dagger
 \bar X^\dagger\Pi_{\rm odd}\bar X
 V_\beta
 }.
 \label{eq:odd-leakage}
\end{equation}
If \(M_4\) and \(M'_4\) agree on the even sectors, then
\begin{equation}
 \norm{
 V_\beta^\dagger\bar X^\dagger
 (M_4-M'_4)\bar X V_\beta
 }
 \le 2\ell_{\rm odd}(\beta).
 \label{eq:extension-bound}
\end{equation}

It is also possible to use one ancilla rather than resolving both modular
bits.  The controlled-quarter-rotation/Hadamard-test idea and its
modular-number interpretation already appear in the original GKP
proposal~\cite{GKP2001}.  In the present phase convention, the test has
oscillator effects
\begin{equation}
 E_x
 =
 \frac12\Id+\frac{x}{4}(F+F^\dagger),
 \qquad x\in\{+1,-1\},
 \label{eq:one-bit-effects}
\end{equation}
and effective observable
\begin{equation}
 C_4:=E_+-E_-=\frac{F+F^\dagger}{2}=\Pi_0-\Pi_2.
 \label{eq:C4}
\end{equation}
This is a valid binary POVM, with random outcomes on the odd sectors, but it
is not the deterministic reflection in Eq.~\eqref{eq:M4}.  Its ancilla
implementation is a projective Naimark dilation and is therefore admissible
in a CHSH experiment.  Our numerical comparison below keeps the two
implementations distinct.  A two-bit phase-estimation measurement that
resolves \(\{\Pi_r\}\), followed by the declared classical coarse-graining,
implements \(M_4\) itself.

\subsection{Finite-energy calibration of the displaced modular measurement}
\label{subsec:calibration-definition}

The canonical setting in Eq.~\eqref{eq:physical-tilted} uses the ideal logical
displacement \(\bar X=D(\sqrt\pi)\).  At finite energy, it is natural to ask
whether its amplitude remains optimal while every other element of the
protocol is kept fixed.  We therefore introduce
\begin{equation}
 D(d):=e^{-id\hat p},
 \qquad
 A_-^{\rm phys}(d):=D(d)^\dagger M_4D(d),
 \qquad d>0.
 \label{eq:displacement-family}
\end{equation}
For every \(d\),
\begin{equation}
 A_-^{\rm phys}(d)^\dagger=A_-^{\rm phys}(d),
 \qquad
 A_-^{\rm phys}(d)^2=\Id,
 \label{eq:calibrated-reflection}
\end{equation}
so the CHSH and self-testing statements apply without modification.  Only at
\(d=\sqrt\pi\) do we call \(D(d)\) the ideal logical Pauli displacement.

For a fixed value of \(\beta\), define
\begin{equation}
 S(\beta,d)
 :=
 \left\langle
 M_4\otimes(Q+P)
 +A_-^{\rm phys}(d)\otimes(Q-P)
 \right\rangle_\beta.
 \label{eq:chsh-distance}
\end{equation}
We define a local calibration around the ideal displacement, not a global optimization over translations.  The same symmetric \(\pm50\%\) window about \(\sqrt\pi\) is fixed for every \(\beta\); it excludes the near-zero duplication of the undisplaced setting and does not extend to multiple neighboring translations.  Thus
\begin{equation}\mathcal I:=\left[\frac12\sqrt\pi,\frac32\sqrt\pi\right].\label{eq:calibration-interval}\end{equation}
This is a protocol definition for a one-parameter finite-energy correction, not a claim of a unique physical interval.
The finite-energy calibrated prescription is
\begin{equation}
 S_{\rm cal}(\beta):=\max_{d\in\mathcal I}S(\beta,d),
 \qquad
 d_{\rm opt}(\beta)\in
 \operatorname*{arg\,max}_{d\in\mathcal I}S(\beta,d).
 \label{eq:calibrated-score}
\end{equation}
The canonical ideal-logical prescription is
\(S_{\rm fix}(\beta):=S(\beta,\sqrt\pi)\).  In the calibration, the binning
\(b(q)=\sgn[\cos(\sqrt\pi q)]\), the filtered Bell state, every other CHSH
setting, and the primary odd-sector assignment
\((s_1,s_3)=(+1,-1)\) remain fixed.  Thus the calibration varies one physical
parameter and does not jointly optimize the odd-sector rule.  The value of
\(d\) is chosen from an independently characterized \(\beta\) before Bell-test
data are collected; optimizing it on the same experimental sample would
require a selection-aware confidence analysis.  The characterization of
\(\beta\) is used only to choose a measurement setting expected to give a
large score; it is not an assumption in the device-independent inference.  If
\(\beta\) is misestimated, or if the selected \(d\) is not optimal, the
observed score may decrease, but any valid confidence lower bound on that
score can still be inserted into Corollary~\ref{thm:kaniewski}.  Calibration data
and Bell-test data should therefore be treated as independent data sets unless
the statistical analysis explicitly accounts for the selection of \(d\).

For comparison, the contracted lattice spacing suggests the diagnostic
\begin{equation}
 d_{\rm ctr}(\beta):=\frac{\sqrt\pi}{\cosh\beta}.
 \label{eq:contracted-distance}
\end{equation}
We do not assume that \(d_{\rm ctr}\) is optimal.  The numerical results below
in fact find \(d_{\rm opt}>\sqrt\pi>d_{\rm ctr}\) over the displayed range,
with \(d_{\rm opt}/\sqrt\pi\) approaching one as the squeezing increases.

\section{Exact physical correlations}
\label{sec:correlations}

\subsection{A compact Gram-matrix formula}

For a physical oscillator operator \(M\), define its covariant matrix in the
unorthonormalized filtered columns by
\begin{equation}
 \Gamma_\beta(M):=W_\beta^\dagger M W_\beta.
 \label{eq:Gamma}
\end{equation}
The expectation of any product observable in the physical state
\(\ket{\Psi_\beta}\) is exactly
\begin{equation}
 \langle M\otimes N\rangle_\beta
 =
 \frac{
 \bra{\psi_{\rm B}}
 \Gamma_\beta(M)\otimes\Gamma_\beta(N)
 \ket{\psi_{\rm B}}
 }{
 \bra{\psi_{\rm B}}
 S_\beta\otimes S_\beta
 \ket{\psi_{\rm B}}
 }.
 \label{eq:physical-correlation-formula}
\end{equation}
Equation~\eqref{eq:physical-correlation-formula} retains both the finite
codeword overlap and the logical distortion induced by the filter.  It is
merely an exact way to evaluate a physical expectation value; no measurement
is replaced by its logical compression.

The rotation covariance reduces the \(P\) calculation to the \(Q\)
calculation:
\begin{equation}
 \Gamma_\beta(P)
 =
 H\Gamma_\beta(Q)H.
 \label{eq:P-from-Q}
\end{equation}
For \(M_4\), even parity and Eq.~\eqref{eq:finite-H-covariance} give
\begin{equation}
 \Gamma_\beta(M_4)
 =
 S_\beta H
 =
 HS_\beta.
 \label{eq:Gamma-M4}
\end{equation}
No analogous simplification removes the finite-energy calculation for
\(D(d)^\dagger M_4D(d)\).

\subsection{Correlators}

We use the following directly measurable quantities:
\begin{align}
 C_{XZ}(\beta)
 &:=
 \langle P\otimes Q\rangle_\beta,
 &
 C_{ZX}(\beta)
 &:=
 \langle Q\otimes P\rangle_\beta,
 \label{eq:cross-correlations}\\
 C_0(\beta)
 &:=
 \langle P\otimes P+Q\otimes Q\rangle_\beta,
 \label{eq:C0}\\
 C_+(\beta)
 &:=
 \left\langle
 M_4\otimes(Q+P)
 \right\rangle_\beta,
 \label{eq:Cplus}\\
 C_-(\beta,d)
 &:=
 \left\langle
 D(d)^\dagger M_4D(d)
 \otimes(Q-P)
 \right\rangle_\beta.
 \label{eq:Cminus}
\end{align}
The CHSH expectation and its fixed-displacement specialization are
\begin{equation}
 S(\beta,d)=C_+(\beta)+C_-(\beta,d),
 \qquad
 S_{\rm fix}(\beta)=S(\beta,\sqrt\pi).
 \label{eq:physical-chsh}
\end{equation}
In the ideal limit, for \(d=\sqrt\pi\),
\[
 C_{XZ},C_{ZX}\to1,
 \quad
 C_0\to0,
 \quad
 C_+,C_-(\,\cdot\,,\sqrt\pi)\to\sqrt2,
 \quad
 S_{\rm fix}\to2\sqrt2.
\]
The algebraic identity
\(2\sqrt2-S_{\rm fix}
=(\sqrt2-C_+)+(\sqrt2-C_-(\beta,\sqrt\pi))\) shows that the aggregate CHSH
shortfall is not independent of the two tilted-setting shortfalls.

\subsection{Fock-space representation}

Let
\begin{equation}
 \phi_n(x)
 =
 \frac{H_n(x)e^{-x^2/2}}
 {\pi^{1/4}\sqrt{2^n n!}}
 \label{eq:hermite-functions}
\end{equation}
be the real harmonic-oscillator wavefunctions.  Up to one common
\(\beta\)-dependent scalar, the Fock coefficients of the filtered columns are
\begin{equation}
 \langle n|w_{j,\beta}\rangle
 =
 e^{-\beta n}
 \sum_{s\in\mathbb Z}
 \phi_n\!\left((2s+j)\sqrt\pi\right).
 \label{eq:fock-coefficients}
\end{equation}
This expression follows directly by inserting the number basis between
\(E_\beta\) and the ideal comb.  It provides an efficient representation of
\(S_\beta\), \(M_4\), and the displaced setting.  The \(Q\) matrix elements
are evaluated by binwise quadrature in position space, and
Eq.~\eqref{eq:P-from-Q} supplies \(P\).  The displacement is applied as the
matrix exponential of the tridiagonal Fock representation of \(\hat p\).
Appendix~\ref{app:numerics} records the numerical protocol in more detail.

For later reporting define
\begin{align}
 \bar n(\beta)&:=\langle\Psi_\beta|\hat n\otimes\Id|\Psi_\beta\rangle=\langle\Psi_\beta|\Id\otimes\hat n|\Psi_\beta\rangle,\label{eq:mean-photon-number}\\
 p_{\rm odd}(\beta,d)&:=\langle\Psi_\beta|D(d)^\dagger\Pi_{\rm odd}D(d)\otimes\Id|\Psi_\beta\rangle.\label{eq:odd-probability}
\end{align}
The first equality follows from exchange symmetry.  We use \(p_{\rm odd}^{\rm fix}=p_{\rm odd}(\beta,\sqrt\pi)\) and \(p_{\rm odd}^{\rm cal}=p_{\rm odd}(\beta,d_{\rm opt}(\beta))\); these probabilities differ from \(\ell_{\rm odd}\) in Eq.~\eqref{eq:odd-leakage}.

\section{Numerical results}
\label{sec:numerics}

\subsection{Fixed and calibrated prescriptions}

For the production data, we used a Fock cutoff \(N_{\rm F}=520\), lattice
indices \(|s|\le22\), position range \(|q|\le30\), and 128-point
Gauss--Legendre quadrature separately on every bin interval.  All codeword
normalizations and overlaps were computed from the same truncated vectors.
The canonical calculation uses precisely Eqs.~\eqref{eq:QP} and
\eqref{eq:physical-tilted}.  The calibrated calculation changes only the
displacement amplitude in Eq.~\eqref{eq:displacement-family}.

For the centered state, the four directly evaluated product correlators obey
\begin{equation}
 C_{00}=C_{01}=\frac{C_+}{2}
 \label{eq:correlator-symmetry}
\end{equation}
to numerical precision, where
\(C_{xy}:=\langle A_x^{\rm phys}\otimes B_y^{\rm phys}\rangle_\beta\),
\((B_0^{\rm phys},B_1^{\rm phys})=(Q,P)\), and
\((A_0^{\rm phys},A_1^{\rm phys})=(M_4,A_-^{\rm phys}(d))\).
The odd probability \(p_{\rm odd}\) is evaluated before binary
coarse-graining; no odd event is postselected.

For each \(\beta\), we first scan 81 equally spaced values on the interval
\(\mathcal I\) in Eq.~\eqref{eq:calibration-interval} and then use a bounded
one-dimensional refinement around the best grid point with relative
displacement tolerance below \(10^{-10}\).  The selected maximum is well
inside the interval and lies on a single continuous branch over the displayed
range.  Additional coarse scans over
\(0.05\le d/\sqrt\pi\le8\) at the calibrated Bell threshold, \(5\) dB, and
\(14\) dB found no larger peak.  We nevertheless define the calibrated
prescription only on the explicit interval \(\mathcal I\), rather than claim a
global maximum over all \(d>0\).

\begin{table}[H]
\centering
\caption{Representative results for the assumed filtered source model.
Here \(\bar n\) is the per-mode expectation in Eq.~\eqref{eq:mean-photon-number}, and \(\fK\) denotes the
model evaluation of Eq.~\eqref{eq:fK} on the computed score.  The simulation
does not provide a certified one-sided numerical error interval.}
\label{tab:primary-results}
\small
\resizebox{\textwidth}{!}{%
\begin{tabular}{ccccccccc}
\toprule
\(r_{\rm dB}\) & \(\bar n\) & \(d_{\rm opt}/\sqrt\pi\) &
\(S_{\rm fix}\) & \(S_{\rm cal}\) &
\(\fK(S_{\rm fix})\) & \(\fK(S_{\rm cal})\) &
\(p_{\rm odd}^{\rm fix}\) & \(p_{\rm odd}^{\rm cal}\)\\
\midrule
5  & 0.8923  & 1.10815 & 2.10052 & 2.21292 & 0.50000 & 0.57410 & 0.32850 & 0.36326\\
6  & 1.4018  & 1.10160 & 2.28636 & 2.41883 & 0.62492 & 0.71658 & 0.27904 & 0.32061\\
8  & 2.6487  & 1.08373 & 2.50604 & 2.64492 & 0.77693 & 0.87303 & 0.19718 & 0.24725\\
10 & 4.5000  & 1.06426 & 2.62198 & 2.74230 & 0.85715 & 0.94041 & 0.13509 & 0.18591\\
12 & 7.4245  & 1.04670 & 2.69486 & 2.78858 & 0.90758 & 0.97243 & 0.08999 & 0.13556\\
14 & 12.0594 & 1.03255 & 2.74267 & 2.81072 & 0.94066 & 0.98775 & 0.05880 & 0.09584\\
\bottomrule
\end{tabular}}
\end{table}

Figure~\ref{fig:main-results} displays the full numerical curves.  The first
tilted contribution \(C_+\) converges more rapidly than the second.  This is
the numerical manifestation of Proposition~\ref{prop:tilted-action}:
\(M_4\) is an exact finite-energy Hadamard in the canonical frame, whereas the
displaced conjugate is not an exact finite-energy \(A_-\).  Calibration raises
\(C_-\), but also raises \(p_{\rm odd}\); the gain therefore relies on the
fixed phase-bit assignment on the odd sectors.  This dependence is part of
the measurement specification, not an additional fitting parameter: the same
assignment is used for every squeezing value and is fixed before the
displacement is calibrated.  The optimized displacement is larger than
\(\sqrt\pi\), while the contracted-spacing diagnostic \(d_{\rm ctr}\) is
smaller, so the two do not coincide.  Thus the optimum is not determined by
the contraction of the lattice centers alone; it reflects the full signed
contribution of the displaced modular sectors, including the odd sectors
activated by the displacement.  The large calibrated gain is a property of
the resolved modular-four PVM followed by the declared deterministic
phase-bit coarse-graining.  Independently optimizing \(d\) for the one-bit
Hadamard-test POVM over the same interval gives only a much smaller
improvement; see Section~\ref{subsec:extension-sensitivity}.  The phase-bit
assignment is prespecified and never optimized with squeezing.
\begin{figure}[H]
\centering
\includegraphics[width=\textwidth]{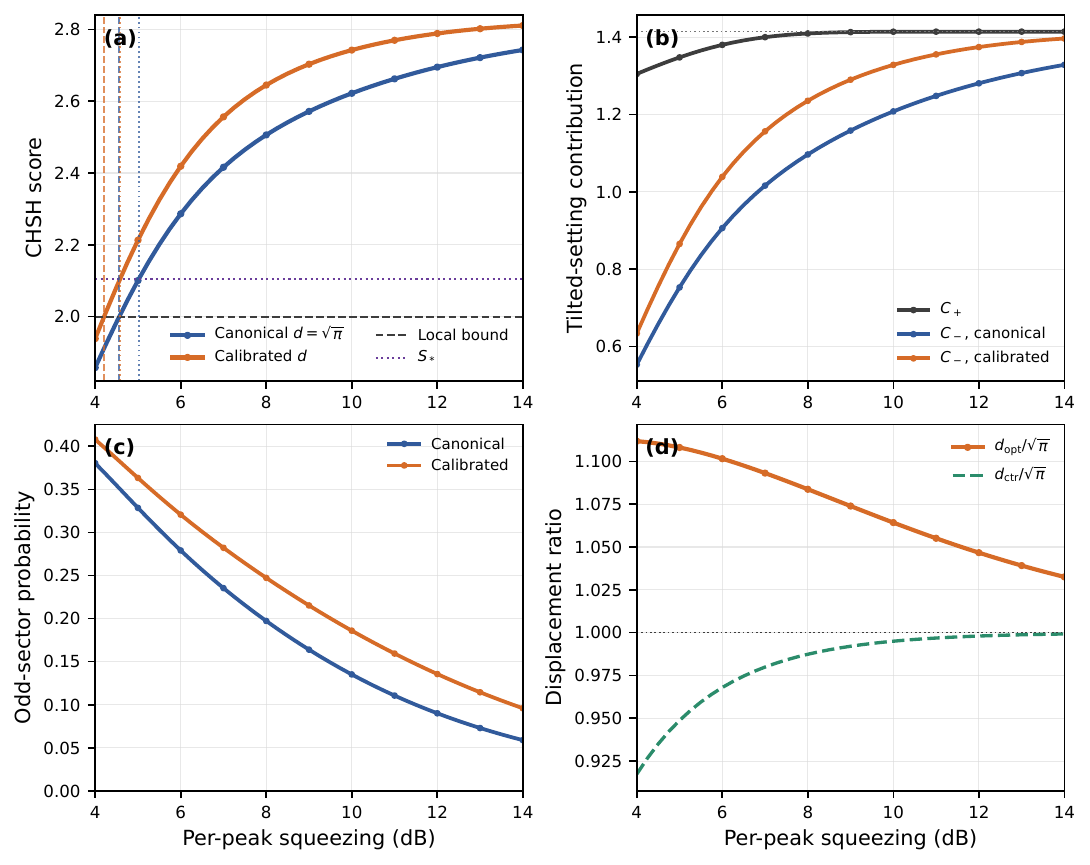}
\caption{Finite-energy performance of the canonical and calibrated
prescriptions.  (a)~CHSH scores, with the local bound and the threshold
\(S_*\), above which the affine branch of the Kaniewski bound exceeds
\(1/2\), marked.  Vertical lines mark the corresponding fixed and calibrated
crossings of the local and \(S_*\) thresholds.  (b)~The two tilted-setting
contributions.  (c)~Odd-sector probabilities after the displacement and before binary coarse-graining.
(d)~The calibrated displacement and \(d_{\rm ctr}\), which is shown only as a
contracted-lattice diagnostic and is not used in the optimization.  Points are
computed at 0.25-dB intervals; connecting lines are guides to the eye.
Calibration varies only \(d\) and keeps \((s_1,s_3)=(+1,-1)\) fixed.  All
calibrated curves shown here refer to this deterministic rule; the
independently calibrated one-bit POVM is not plotted.}
\label{fig:main-results}
\end{figure}

The fixed prescription crosses the local and Kaniewski-bound thresholds at
\begin{align}
 r_{\rm dB}^{\rm Bell,fix}&=4.56\ {\rm dB},
 \label{eq:Bell-threshold}\\
 r_{\rm dB}^{\rm K,fix}&=5.02\ {\rm dB}.
 \label{eq:ST-threshold}
\end{align}
The calibrated prescription crosses them at
\begin{align}
 r_{\rm dB}^{\rm Bell,cal}&=4.21\ {\rm dB},
 \label{eq:Bell-threshold-cal}\\
 r_{\rm dB}^{\rm K,cal}&=4.58\ {\rm dB}.
 \label{eq:ST-threshold-cal}
\end{align}
These are model-based numerical estimates for the state and measurements
defined above, not loss, detector, or finite-sample thresholds for a particular
experiment.  In an experiment, evaluating \(\fK\) at a valid confidence lower
bound \(S_{\rm L}\) gives the device-independent confidence statement in
Corollary~\ref{thm:kaniewski}.

\subsection{The additional \texorpdfstring{\(X/Z\)}{X/Z} checks}

Table~\ref{tab:cross-results} reports the physical correlations from
Eq.~\eqref{eq:extra-ideal-checks}.  The symmetry of the centered construction
gives \(C_{XZ}=C_{ZX}\) to numerical precision.  These quantities approach
their ideal targets more rapidly than the displaced tilted setting.

\begin{table}[H]
\centering
\caption{Physical cross-correlations for the same runs as
Table~\ref{tab:primary-results}.  Values below \(10^{-12}\) are reported only
as a bound consistent with the numerical precision.}
\label{tab:cross-results}
\begin{tabular}{c@{\qquad}c@{\qquad}c}
\toprule
\(r_{\rm dB}\) & \(C_{XZ}=C_{ZX}\) & \(\lvert C_0\rvert\)\\
\midrule
5  & 0.869080 & \(7.799\times10^{-2}\)\\
6  & 0.941142 & \(2.222\times10^{-2}\)\\
7  & 0.977658 & \(4.365\times10^{-3}\)\\
8  & 0.993016 & \(5.546\times10^{-4}\)\\
9  & 0.998293 & \(4.139\times10^{-5}\)\\
10 & 0.999698 & \(1.590\times10^{-6}\)\\
11 & 0.999965 & \(2.650\times10^{-8}\)\\
12 & 0.999998 & \(1.545\times10^{-10}\)\\
13 & 1.000000 & \(<10^{-12}\)\\
14 & 1.000000 & \(<10^{-12}\)\\
\bottomrule
\end{tabular}
\end{table}

\subsection{Cutoff convergence}
\label{subsec:convergence}

The small imaginary part of a computed expectation is not, by itself, an
accuracy test.  We instead varied every truncation parameter independently at
\(14\) dB, the most demanding row in Table~\ref{tab:primary-results}.
Table~\ref{tab:convergence} changes one parameter at a time while holding the
others at the production values
\((N_{\rm F},K,q_{\max},N_{\rm GL})=(520,22,30,128)\).
Here \(K\) is the lattice cutoff, \(q_{\max}\) is the position cutoff, and
\(N_{\rm GL}\) is the Gauss--Legendre order per bin.  Once
\(N_{\rm F}\ge400\), \(K\ge10\), \(q_{\max}\ge28\), and
\(N_{\rm GL}\ge32\), the CHSH score is stable at the \(10^{-9}\) level or
better in these one-parameter tests.  This study is not interval arithmetic
or an analytic tail bound and therefore does not certify the sign of the
remaining truncation error; it is convergence evidence for the reported
model estimates.

\begin{table}[H]
\centering
\caption{One-parameter-at-a-time convergence test at \(14\) dB.  A dash
denotes the production value of the parameter being varied.}
\label{tab:convergence}
\resizebox{\textwidth}{!}{%
\begin{tabular}{cccccc}
\toprule
\(\text{varied parameter}\) & \(\text{value}\) &
\(C_{XZ}\) & \(C_+\) & \(C_-(\beta,\sqrt\pi)\) & \(S_{\rm fix}\)\\
\midrule
production & --  & 0.999999999 & 1.414213561 & 1.328451880 & 2.742665441\\
\addlinespace
\(N_{\rm F}\) & 240 & 0.999999995 & 1.414213559 & 1.328451871 & 2.742665429\\
                & 320 & 0.999999999 & 1.414213561 & 1.328451880 & 2.742665441\\
                & 400 & 0.999999999 & 1.414213561 & 1.328451880 & 2.742665441\\
\addlinespace
\(K\)           & 6   & 0.999999999 & 1.414213561 & 1.328451878 & 2.742665439\\
                & 8   & 0.999999999 & 1.414213561 & 1.328451880 & 2.742665441\\
                & 10  & 0.999999999 & 1.414213561 & 1.328451880 & 2.742665441\\
\addlinespace
\(q_{\max}\)    & 22  & 0.999999998 & 1.414213561 & 1.328451880 & 2.742665441\\
                & 26  & 0.999999999 & 1.414213561 & 1.328451880 & 2.742665441\\
                & 28  & 0.999999999 & 1.414213561 & 1.328451880 & 2.742665441\\
\addlinespace
\(N_{\rm GL}\)  & 16  & 0.999976100 & 1.414196663 & 1.328436006 & 2.742632668\\
                & 24  & 0.999999998 & 1.414213561 & 1.328451880 & 2.742665441\\
                & 32  & 0.999999999 & 1.414213561 & 1.328451880 & 2.742665441\\
\bottomrule
\end{tabular}}
\end{table}

We also varied all cutoffs together and recomputed the four threshold roots.
This directly tests the numerically delicate region rather than inferring
threshold stability from the \(14\)-dB row alone.  Table~\ref{tab:coupled-thresholds}
shows that a coupled reduced calculation, the production calculation, and a
jointly enlarged calculation are stable across the tested truncations far beyond the precision quoted in
Eqs.~\eqref{eq:Bell-threshold}--\eqref{eq:ST-threshold-cal}.  The optimized
displacement ratios at the two calibrated crossings likewise agree within
\(10^{-9}\).  Root finding used absolute dB tolerances below
\(2\times10^{-7}\); changing the calibration grid from 61 to 101 points did
not change any quoted digit.

\begin{table}[H]
\centering
\caption{Coupled-cutoff convergence of the four model thresholds.  The cutoff
tuple is \(N_{\rm F},K,q_{\max},N_{\rm GL}\).  The six-digit values are numerical roots under the stated protocol, not certified interval enclosures.}
\label{tab:coupled-thresholds}
\small
\resizebox{\textwidth}{!}{%
\begin{tabular}{ccccc}
\toprule
cutoffs & \(r_{\rm dB}^{\rm Bell,fix}\) & \(r_{\rm dB}^{\rm K,fix}\) &
\(r_{\rm dB}^{\rm Bell,cal}\) & \(r_{\rm dB}^{\rm K,cal}\)\\
\midrule
\((400,14,28,64)\)   & 4.561246 & 5.024555 & 4.207852 & 4.584106\\
\((520,22,30,128)\)  & 4.561246 & 5.024555 & 4.207852 & 4.584106\\
\((640,26,34,160)\)  & 4.561246 & 5.024555 & 4.207852 & 4.584106\\
\bottomrule
\end{tabular}}
\end{table}

At the production cutoffs, the maximum normalization defect in the four
modular-sector probabilities over the fixed-\(d=\sqrt\pi\), \(5\)--\(14\) dB data set is
below \(10^{-15}\).  The Fourier-covariance residual at \(14\) dB is below
\(3\times10^{-15}\), and all raw expectation values are real to double
precision.  The lattice truncation with \(K=6\) has a visibly larger
Fourier-covariance residual, \(2.1\times10^{-5}\), even though its final score
is already close; this is why the production calculation uses \(K=22\).
Independent Mehler-kernel/Fock calculations of the \(Q\) matrix and
Laguerre-polynomial/matrix-exponential calculations of the displaced setting
were repeated at all four crossings.  Their spectral-norm residuals are at
the \(10^{-14}\) and \(10^{-15}\) levels, respectively; further details and
the machine-readable outputs are given in Appendix~\ref{app:numerics} and the
ancillary data.

\subsection{Sensitivity to the odd-sector convention}
\label{subsec:extension-sensitivity}

The dependence on the deterministic odd-sector assignments is exactly affine.
For \(r=1,3\), define
\begin{equation}
 u_r(\beta,d)
 :=
 \left\langle
 D(d)^\dagger\Pi_rD(d)\otimes(Q-P)
 \right\rangle_\beta.
 \label{eq:odd-affine-coefficients}
\end{equation}
Because the centered state is even before the undisplaced modular measurement,
the score associated with Eq.~\eqref{eq:all-extensions} is
\begin{equation}
 S(s_1,s_3;\beta,d)
 =
 S_{\rm one\text{-}bit}(\beta,d)+s_1u_1(\beta,d)+s_3u_3(\beta,d),
 \label{eq:odd-affine-score}
\end{equation}
where \(S_{\rm one\text{-}bit}\) uses the POVM observable \(C_4=\Pi_0-\Pi_2\).
It follows exactly that the midpoint of the minimum and maximum over the four
deterministic assignments, and also the average of their four scores, equals
\(S_{\rm one\text{-}bit}\).

For a separate comparison, we also define the independently calibrated
one-bit prescription
\begin{equation}
 S_{\rm 1b,cal}(\beta)
 :=
 \max_{d\in\mathcal I}S_{\rm one\text{-}bit}(\beta,d),
 \qquad
 d_{\rm 1b,opt}(\beta)
 \in
 \operatorname*{arg\,max}_{d\in\mathcal I}
 S_{\rm one\text{-}bit}(\beta,d).
 \label{eq:one-bit-calibration}
\end{equation}

\begin{table}[H]
\centering
\caption{CHSH scores for all four deterministic odd-sector assignments and
the one-bit POVM.  ``Fixed'' uses \(d=\sqrt\pi\).  ``At primary
\(d_{\rm opt}\)'' evaluates all four assignments and the one-bit POVM at the
single displacement optimized for the primary \((+1,-1)\) phase-bit rule.
Thus the one-bit entries in those rows are not its independently calibrated
scores.}
\label{tab:extension-sensitivity}
\small
\resizebox{\textwidth}{!}{%
\begin{tabular}{ccccccc}
\toprule
prescription & \(r_{\rm dB}\) & \((-1,-1)\) & \((-1,+1)\) &
\((+1,-1)\) & \((+1,+1)\) & one-bit at same \(d\)\\
\midrule
fixed & 5  & 2.157429 & 2.100523 & 2.100523 & 2.043616 & 2.100523\\
fixed & 8  & 2.514677 & 2.506040 & 2.506040 & 2.497402 & 2.506040\\
fixed & 12 & 2.694862 & 2.694856 & 2.694856 & 2.694849 & 2.694856\\
\addlinespace
at primary \(d_{\rm opt}\) & 5  & 2.022643 & 1.741629 & 2.212917 & 1.931903 & 1.977273\\
at primary \(d_{\rm opt}\) & 8  & 2.384900 & 2.102676 & 2.644922 & 2.362698 & 2.373799\\
at primary \(d_{\rm opt}\) & 12 & 2.608021 & 2.427431 & 2.788579 & 2.607989 & 2.608005\\
at primary \(d_{\rm opt}\) & 14 & 2.678827 & 2.546933 & 2.810721 & 2.678827 & 2.678827\\
\bottomrule
\end{tabular}}
\end{table}

For the canonical displacement, the primary \((+1,-1)\) rule and the one-bit
POVM give the same aggregate CHSH score within \(1.4\times10^{-15}\) at the
tabulated points.  We use this only as a numerical observation: the individual
odd-sector correlators do not vanish, and their cancellation occurs only in
the CHSH functional.  The cancellation disappears away from the canonical
displacement.  At \(5\) dB, for example, the one-bit score evaluated at the
displacement optimized for the primary phase-bit rule is \(1.977273\) and
therefore does not violate CHSH.  This borrowed setting should not be confused
with the independent calibration in Eq.~\eqref{eq:one-bit-calibration}.
Maximizing the one-bit score over the same interval \(\mathcal I\) gives
\(2.105517\) at \(5\) dB, compared with \(2.100523\) at the canonical
displacement and \(2.212917\) for the independently calibrated primary
phase-bit rule.  The independently calibrated one-bit Bell and \(S_*\)
crossings are approximately \(4.537\) dB and \(5.001\) dB, respectively.
Thus one-bit calibration produces only a small improvement, whereas the large
calibrated gain materially uses the deterministic phase-bit rule fixed in
Eq.~\eqref{eq:M4}.  This does not represent a hidden optimization over
odd-sector assignments: \((s_1,s_3)=(+1,-1)\) is declared before calibration,
the same rule is used for every \(\beta\), and no odd event is discarded.

\section{Discussion}
\label{sec:discussion}

\subsection{Relation to prior work and scope of the contribution}

The quarter-rotation identity and its modular-photon-number interpretation are
known from the original GKP proposal, as is diagonal homodyne as a logical
\(Y\) readout~\cite{GKP2001}.  The contribution of the present work is their
full-oscillator finite-energy incorporation into a CHSH self-testing
construction.  We define the binary measurements on every photon-number
sector, evaluate the displaced setting directly without postselection or an
abstract logical correction, retain the finite codeword overlap, and obtain
quantitative Bell-violation and extractability thresholds.  The one-parameter
calibration further shows how the physically consequential odd sectors may be
used without optimizing either the homodyne bins or their deterministic
assignments.

The original GKP discussion also cautions that the ideal
Hadamard/modular-number correspondence need not survive an arbitrary
finite-energy approximation~\cite{GKP2001}.  Our exact statement in
Proposition~\ref{prop:exact-H} is therefore deliberately limited to the
rotationally invariant number filter \(e^{-\beta\hat n}\); it is not asserted
for generic approximate grid states.
Related recent work uses the same rotationally invariant finite-energy and
modulo-four structure for self-Kerr-based GKP magic-state preparation rather
than Bell testing~\cite{Boudreault2026}.

This construction complements two close lines of work.  Marshall and
Weedbrook proposed a GKP-encoded CHSH protocol in which the tilted bases are
reached using a logical \(T\) gate followed by homodyne
readout~\cite{MarshallWeedbrook2014}.  Our route replaces that basis change by
modular photon-number information and a displacement.  Yang \emph{et al.}
showed that the homodyne-only logical-Pauli palette cannot violate CHSH on an
encoded Bell pair~\cite{Yang2026}.  The modular-four setting lies outside that
palette and therefore evades the obstruction.  It does so by introducing a
non-Gaussian measurement resource, not by contradicting the homodyne-only
no-go theorem.
Lopetegui-Gonz\'alez \emph{et al.} take the complementary route of tailoring
continuous-variable states and Bell inequalities to homodyne
measurements~\cite{Lopetegui2026}.  More generally, modular-variable methods
provide a bridge between oscillator observables and discrete quantum
information~\cite{Ketterer2016}.  The individual ingredients are known; the
advance is their quantitative full-oscillator finite-energy combination for
CHSH self-testing.

\subsection{Experimental implementation and imperfections}

The \(Q/P\) settings require homodyne readout and deterministic periodic
post-processing.  The projective \(M_4\) setting requires resolution of
\(\hat n\bmod4\), for example by a two-bit phase-estimation circuit using
controlled \(F\) and controlled \(F^2=(-1)^{\hat n}\), followed by the fixed
coarse-graining in Eq.~\eqref{eq:M4}.  In a cavity architecture these
controlled rotations can in principle be generated dispersively
~\cite{TerhalWeigand2016}.  The \(A_-^{\rm phys}\) setting additionally
applies the Gaussian displacement \(\bar X\) before the modular measurement.
If only the binary logical outcome is required, the one-bit alternative in
Eq.~\eqref{eq:C4} avoids deterministic assignment of the odd sectors but
implements a different POVM.

The present numerics isolate the intrinsic effect of finite energy.  They do
not model a specific detector efficiency, loss channel, ancilla fault model,
or finite sample size.  This is a deliberate scope boundary rather than an
assumption in the eventual device-independent conclusion.  All such effects
change the observed CHSH score, and Corollary~\ref{thm:kaniewski} applies to a
valid lower confidence bound on that observed score.  For example, if the four
product correlators in the expanded CHSH expression each differ from the
model by at most \(\eta\), then the observed score is lower bounded by
\(S(\beta,d)-4\eta\).

Photon loss is particularly relevant for the modular-four settings.  For an input component in an even residue sector, a single annihilation event moves it to the adjacent odd residue sector with the opposite \(M_4\) assignment.  Components already in an odd sector need not undergo a binary-value flip, so a complete loss analysis must retain the full residue distribution.  Loss during a
controlled rotation can also correlate the error with the measurement
ancilla.  Quantitative loss and fault-tolerance thresholds therefore require
a platform-specific open-system analysis; the finite-energy thresholds in
Eqs.~\eqref{eq:Bell-threshold}--\eqref{eq:ST-threshold-cal} should not be read
as experimental loss budgets.

\subsection{Certification scope}

The explicit numerical conclusion of this work is a predicted CHSH score and
a model-based evaluation of the state-extractability curve.  If instead a
valid experimental confidence lower bound \(S_{\rm L}\) is supplied,
Corollary~\ref{thm:kaniewski} gives the extractability bound
\(\Xi(\rho\to\psi_{\rm B})\ge\fK(S_{\rm L})\).  Equivalently, for every
\(\eps>0\), there are local normal CPTP maps whose qubit outputs have overlap
at least \(\fK(S_{\rm L})-\eps\) with \(\ket{\psi_{\rm B}}\).  We do not claim
from that corollary a separate finite-error operator-norm certification of all
six measurements.  The additional \(X/Z\) checks in
Table~\ref{tab:cross-results} are reported because they are useful in the
larger verification setting of Ref.~\cite{HayashiHajdusek2018}; they are not
presented as independent inputs to the Kaniewski fidelity bound.
The additional six-setting correlations reported here may also serve as
input to extensions of the self-guaranteed measurement-based verification
framework of Hayashi and Hajdu\v{s}ek to finite-energy GKP graph states; such
a graph-state result is not proved in the present paper.

Likewise, calculating the honest-device correlations is not itself a
device-independent experiment.  If black-box devices produce an observed
score, the extraction statement depends only on that score and not on the
GKP interpretation.  The GKP analysis establishes that the proposed
source-and-measurement model can enter the nontrivial regime.
The scope of robustness is the score-to-extractability relation stated in the Introduction.

\subsection{Other GKP lattices}

Bare photon number modulo four is tied to the centered square oscillator
frame.  At the ideal level, if a one-mode lattice is obtained by a symplectic
map \(S\) with Gaussian unitary \(U_S\), then the transported measurement is
\[
 M_{4,S}=U_SM_4U_S^\dagger.
\]
For a displaced phase frame, the corresponding displacement must also
conjugate the measurement.  At finite energy one must distinguish the
transported family \(U_SE_\beta\cC_{\square}\) from the native family
\(E_\beta U_S\cC_{\square}\), because a squeezing component of \(U_S\)
does not commute with \(E_\beta\).  Thus symplectic transport gives an exact
statement for the transported family but not automatically for every native
finite-energy regularization~\cite{Conrad2022}.  A quantitative nonsquare
analysis is outside the scope of this paper.

\section{Conclusion}
\label{sec:conclusion}

We construct four explicit full-oscillator measurements for a finite-energy
GKP CHSH test.  Position and momentum binnings supply the Pauli settings.  A
fixed binary coarse-graining of photon number modulo four supplies the first
tilted setting, and its displaced conjugate supplies the second.  Every
measurement is defined on the full Hilbert space, no outcome is postselected,
and all honest-model correlations are evaluated in the same normalized
number-filtered Bell state.

For the canonical displacement, the model predicts CHSH violation above
\(4.56\) dB and a nontrivial Kaniewski-bound fidelity estimate above
\(5.02\) dB.  Calibrating only the displacement amplitude lowers these
crossings to \(4.21\) dB and \(4.58\) dB.  At \(12\) dB, the score increases
from \(2.69486\) to \(2.78858\), and the corresponding model evaluation of
the fidelity bound increases from \(0.90758\) to \(0.97243\).  The optimized
amplitude approaches the ideal value as the squeezing increases.

The central finite-energy effect is parity leakage: the displacement does not
preserve the centered filtered code manifold and activates the odd
modulo-four sectors.  Their a priori binary assignments are therefore part of
the physical measurement, not an arbitrary detail that may be ignored after
logical compression.  Modular photon-number information supplies the tilted
axes unavailable to the homodyne-only Pauli palette, and the resulting scores
enter the regime where observed black-box correlations yield a quantitative,
dimension-independent Bell-pair extractability guarantee.  The corollary
certifies an extracted Bell pair from the observed score; it does not by
itself certify a GKP encoding or the detailed oscillator realization.

\appendix

\section{Fixing the phase of the quarter rotation}
\label{app:phase}

It is useful to verify Eq.~\eqref{eq:F-logical-H} directly because the global
phase of a measured unitary fixes its binary outcome labels.  With
\[
 {}_q\langle q|p\rangle
 =
 \frac{e^{iqp}}{\sqrt{2\pi}},
 \qquad
 F\ket{x}_q=\ket{x}_p,
\]
the position wavefunction of the transformed ideal codeword is
\begin{align}
 {}_q\langle q|F|\bar j\rangle
 &\propto
 \frac{1}{\sqrt{2\pi}}
 \sum_{s\in\mathbb Z}
 e^{iq(2s+j)\sqrt\pi}
 \nonumber\\
 &=
 \frac{e^{ij\sqrt\pi q}}{\sqrt{2\pi}}
 \sum_{s\in\mathbb Z}e^{i2s\sqrt\pi q}.
 \label{eq:poisson-start}
\end{align}
Poisson summation turns the last sum into a Dirac comb supported at
\(q=k\sqrt\pi\).  At those points,
\(e^{ij\sqrt\pi q}=(-1)^{jk}\).  Splitting \(k\) into its even and odd
parts therefore gives, with the common ideal normalization fixed,
\[
 F\ket{\bar j}
 =
 \frac{\ket{\bar0}+(-1)^j\ket{\bar1}}{\sqrt2}.
\]
This is \(+H\).  By contrast, the metaplectic operator
\(R(\pi/2)=e^{i\pi/4}F\) carries the extra phase \(e^{i\pi/4}\).
The latter is harmless for conjugation identities but not when the unitary
itself is read out as an observable.

\section{Extension of the CHSH bound by finite-rank approximation}
\label{app:infinite-dimension}

\begin{proof}[Proof of Corollary~\ref{thm:kaniewski}]
We work only with the abstract state and binary measurements in the CHSH
experiment.  No oscillator code subspace or decomposition into
two-dimensional sectors is used.  Let \(E^A_x\) and \(E^B_y\) denote one
effect of each binary POVM and set
\[
 A_x=2E^A_x-\Id_A,
 \qquad
 B_y=2E^B_y-\Id_B,
 \qquad x,y\in\{0,1\}.
\]
These operators are Hermitian contractions.  In this proof write
\[
 \mathsf W
 :=
 A_0\otimes(B_0+B_1)+A_1\otimes(B_0-B_1),
 \qquad
 S=\Tr(\rho\mathsf W).
\]
Choose increasing finite-rank
projections \(P_n\to\Id_A\) and \(Q_n\to\Id_B\) strongly, and write
\[
 R_n=P_n\otimes Q_n,
 \qquad
 p_n=\Tr(R_n\rho),
 \qquad
 \rho_n=\frac{R_n\rho R_n}{p_n}.
\]
For all sufficiently large \(n\), \(p_n>0\).  Since \(\rho\) is trace class
and \(R_n\to\Id_A\otimes\Id_B\) strongly,
\begin{equation}
 p_n\longrightarrow1,
 \qquad
 \norm{\rho_n-\rho}_1
 \le 2\sqrt{1-p_n}+(1-p_n)\longrightarrow0,
 \qquad
 \delta_n:=\frac12\norm{\rho_n-\rho}_1\longrightarrow0.
 \label{eq:finite-rank-state-convergence}
\end{equation}
Here the displayed estimate combines the gentle-measurement bound with the
normalization term
\(\norm{R_n\rho R_n-\rho_n}_1=1-p_n\).

On \(P_n\cH_A\) and \(Q_n\cH_B\), respectively, define the compressed
effects
\[
 E^{A,(n)}_x=P_nE^A_xP_n,
 \qquad
 E^{B,(n)}_y=Q_nE^B_yQ_n.
\]
Together with their complements relative to \(P_n\) and \(Q_n\), these are
valid binary POVMs on the finite-dimensional compressed spaces.  Their
associated observables are
\[
 A_x^{(n)}=P_nA_xP_n,
 \qquad
 B_y^{(n)}=Q_nB_yQ_n.
\]
If \(\mathsf W_n\) is the resulting CHSH operator, then, under its natural
embedding into the full tensor-product space,
\[
 \mathsf W_n=R_n\mathsf W R_n
 \quad\hbox{on }R_n(\cH_A\otimes\cH_B).
\]
Consequently,
\begin{equation}
 S_n
 :=
 \Tr(\rho_n\mathsf W_n)
 =\Tr(\rho_n\mathsf W),
 \qquad
 |S_n-S|
 \le\norm{\mathsf W}_\infty\norm{\rho_n-\rho}_1
 \longrightarrow0.
 \label{eq:truncated-score-convergence}
\end{equation}
The equality uses \(\rho_n=R_n\rho_nR_n\); moreover,
\(\norm{\mathsf W}_\infty\le2\sqrt2\) by the Tsirelson bound for Hermitian
contractions.

If \(S=2\), constant local channels already give overlap
\(1/2=\fK(2)\).  Suppose therefore that \(S>2\).  Then
\(S_n\in[2,2\sqrt2]\) for all sufficiently large \(n\).  The upper bound is
the usual Tsirelson bound for Hermitian contractions.  Theorem
\ref{thm:kaniewski-finite} supplies local channels
\[
 \Lambda_{A,n}:\mathcal T(P_n\cH_A)\longrightarrow
                 \mathcal T(\mathbb C^2),
 \qquad
 \Lambda_{B,n}:\mathcal T(Q_n\cH_B)\longrightarrow
                 \mathcal T(\mathbb C^2)
\]
such that
\begin{equation}
 \bra{\psi_{\rm B}}
 (\Lambda_{A,n}\otimes\Lambda_{B,n})(\rho_n)
 \ket{\psi_{\rm B}}
 \ge \fK(S_n).
 \label{eq:finite-rank-kaniewski}
\end{equation}
For clarity, if the finite-dimensional result is expressed for reflections,
the compressed POVMs are first dilated locally.  For any effect
\(0\le E\le\Id\), on \(\cH\oplus\cH\) set
\[
 R_E=
 \begin{pmatrix}
  2E-\Id & 2\sqrt{E(\Id-E)}\\
  2\sqrt{E(\Id-E)} & \Id-2E
 \end{pmatrix},
 \qquad
 V\xi=\xi\oplus0.
\]
Functional calculus gives \(R_E=R_E^\dagger\), \(R_E^2=\Id\), and
\(V^\dagger R_EV=2E-\Id\).  The same isometry \(V\) is used for both
settings at a party, so every CHSH correlator is preserved.  Composing the
extraction channel on the dilation with \(V\) gives the channels in
Eq.~\eqref{eq:finite-rank-kaniewski}.

Choose fixed qubit states \(\tau_A,\tau_B\).  Extend the truncated channels
to the full local trace classes by
\begin{align}
 \widetilde\Lambda_{A,n}(T)
 &=
 \Lambda_{A,n}(P_nTP_n)
 +\Tr[(\Id_A-P_n)T]\,\tau_A,
 \nonumber\\
 \widetilde\Lambda_{B,n}(T)
 &=
 \Lambda_{B,n}(Q_nTQ_n)
 +\Tr[(\Id_B-Q_n)T]\,\tau_B.
 \label{eq:extended-truncated-channels}
\end{align}
These maps are normal and CPTP.  For example, the Heisenberg adjoint of the
first is
\[
 \widetilde\Lambda_{A,n}^{\dagger}(C)
 =
 P_n\Lambda_{A,n}^{\dagger}(C)P_n
 +\Tr(\tau_A C)(\Id_A-P_n),
\]
which is unital and completely positive; normality is immediate because the
domain of this Heisenberg adjoint is finite-dimensional.
The spatial tensor product
\(\widetilde\Lambda_{A,n}\otimes\widetilde\Lambda_{B,n}\) is therefore itself
a normal CPTP map on the bipartite trace class.

Let \(\Pi_{\rm B}=\ketbra{\psi_{\rm B}}{\psi_{\rm B}}\).  Trace-distance
contractivity and \(0\le\Pi_{\rm B}\le\Id\) give the following.  Because
\(\rho_n\) is supported on
\(P_n\cH_A\otimes Q_n\cH_B\), the extended product channel agrees on
\(\rho_n\) with \(\Lambda_{A,n}\otimes\Lambda_{B,n}\):
\begin{align}
 &\Tr\!\left[
 \Pi_{\rm B}
 (\widetilde\Lambda_{A,n}\otimes
  \widetilde\Lambda_{B,n})(\rho)
 \right]
 \nonumber\\
 &\qquad\ge
 \Tr\!\left[
 \Pi_{\rm B}
 (\widetilde\Lambda_{A,n}\otimes
  \widetilde\Lambda_{B,n})(\rho_n)
 \right]-\delta_n
 \nonumber\\
 &\qquad\ge \fK(S_n)-\delta_n.
 \label{eq:full-state-approximate-extraction}
\end{align}
For every such \(n\), \(\Xi\) dominates the left-hand side of
Eq.~\eqref{eq:full-state-approximate-extraction}, and hence it also dominates
its limit superior.  Using Eqs.~\eqref{eq:finite-rank-state-convergence}
and~\eqref{eq:truncated-score-convergence} and continuity of \(\fK\) gives
\[
 \Xi(\rho\to\psi_{\rm B})\ge\fK(S).
\]

The channel pair used in this estimate may depend on \(n\).  No convergence
of the channels and no maximizing channel are required: by definition,
\(\Xi\) is the supremum over all local channel pairs, so it is at least the
value achieved by each pair
\((\widetilde\Lambda_{A,n},\widetilde\Lambda_{B,n})\).  Taking the limit of
the resulting scalar lower bounds therefore proves
Eq.~\eqref{eq:kaniewski-separable-bound}.  Equivalently, choosing a sufficiently
large \(n\) yields, for every \(\eps>0\), a channel pair with overlap at least
\(\fK(S)-\eps\).

The projections and normalized states used above occur only inside this
approximation proof.  They neither modify nor postselect the physical Bell
experiment, and the resulting extraction channels output abstract qubit
registers rather than selecting a physical two-dimensional oscillator sector.
\end{proof}

\section{Numerical details}
\label{app:numerics}

\subsection{Truncated filtered columns}

For \(n=0,\ldots,N_{\rm F}-1\), the code evaluates
Eq.~\eqref{eq:fock-coefficients} using the stable recurrence
\begin{align}
 \phi_0(x)&=\pi^{-1/4}e^{-x^2/2},\\
 \phi_1(x)&=\sqrt2\,x\phi_0(x),\\
 \phi_{n+1}(x)
 &=
 \sqrt{\frac{2}{n+1}}\,x\phi_n(x)
 -
 \sqrt{\frac{n}{n+1}}\,\phi_{n-1}(x).
 \label{eq:hermite-recurrence}
\end{align}
The same truncated columns determine \(S_\beta\), the state normalization,
the mean photon number, and all Fock-diagonal observables.  The canonical
isometry is formed by diagonalizing the positive \(2\times2\) Gram matrix,
not by normalizing the two columns separately.

\subsection{Quadrature and displacement calculations}

The position wavefunctions of the two canonical columns are reconstructed
from Eq.~\eqref{eq:hermite-recurrence}.  The integration domain is split at
every discontinuity of \(b(q)\), and each open bin interval is integrated by
Gauss--Legendre quadrature.  This avoids assigning a finite quadrature weight
to a bin boundary.  The momentum matrix follows from
Eq.~\eqref{eq:P-from-Q}.
The quadrature is first evaluated in the canonical columns.  The covariant
matrix used in the physical filtered-state correlation is then reconstructed
with both Gram factors:
\begin{equation}
 \Gamma_\beta(Q)
 =
 S_\beta^{1/2}
 \bigl(V_\beta^\dagger QV_\beta\bigr)
 S_\beta^{1/2}.
 \label{eq:restore-Gram-Q}
\end{equation}
Thus the numerical integration computes the quantity in
Eq.~\eqref{eq:Gamma}, not a correlation in the canonically embedded state.

In the Fock basis,
\[
 \hat p=\frac{\hat a-\hat a^\dagger}{i\sqrt2}
\]
is tridiagonal.  We apply
\(D(d)=e^{-id\hat p}\) to the two filtered columns through the
action of its matrix exponential and then insert the diagonal residue-class
operator \(M_4\).  For each parameter set we check:
\begin{align}
 &\norm{V_\beta^\dagger V_\beta-\Id_2},\\
 &\norm{FV_\beta-V_\beta H},\\
 &\norm{V_\beta^\dagger M_4V_\beta-H},\\
 &\left|\sum_{r=0}^3p_r-1\right|.
\end{align}
In the production runs these checks are at or below the observed numerical
truncation sensitivity.  The independent representation checks were carried
out at the four threshold points, with the calibrated displacement used at
the calibrated crossings, and at \(14\) dB.  Across these checks, evaluating
the \(Q\) matrix from the exact Mehler-kernel Gaussian comb agrees with the
Fock reconstruction at spectral-norm residuals below \(10^{-13}\).
Analytic Laguerre-polynomial matrix elements of the displacement reproduce
the exponentiated-Fock displaced-setting matrix at residuals below
\(10^{-14}\).  The source code, frozen data, and individual residuals are
included in the ancillary reproducibility package.

\section*{Data Availability Statement}

The numerical source code and machine-readable data underlying the tables and
Fig.~\ref{fig:main-results} are included in the ancillary reproducibility
archive supplied with this manuscript.  No experimental data were generated.

\section*{Acknowledgements} 
F.S.'s work is supported by the European Commission through a Marie
Sk{\l}odowska-Curie Global Fellowship.  He gratefully acknowledges the
hospitality of the Institute for Quantum Computing (IQC), University of
Waterloo, and Freie Universit\"at Berlin.  Part of this work was carried out
while F.S. was supported as a Walter Benjamin Fellow by the Deutsche
Forschungsgemeinschaft (DFG) at the Technical University of Munich.
M.H. was supported in part by the General R\&D Projects of 1+1+1
CUHK--CUHK(SZ)--GDST Joint Collaboration Fund (Grant No.~GRDP2025-022), the
Guangdong Provincial Quantum Science Strategic Initiative (Grant
No.~GDZX2505003), and the Shenzhen International Quantum Academy (Grant
No.~SIQA2025KFKT07).
During the preparation of this manuscript, the authors used Microsoft Copilot with the GPT-5.5 and 5.6 Thinking model to assist with language editing, organization and presentation of the manuscript, and the exploration, development, and checking of certain mathematical derivations and arguments. 
All AI-assisted material was critically reviewed, verified, and revised by the authors, who take full responsibility for the accuracy and integrity of the manuscript.

\sloppy

\fussy

\end{document}